\pdfoutput=1
\RequirePackage{etoolbox}
\csdef{input@path}{%
 {sty/}
 {img/}
}%
\csgdef{bibdir}{}

\documentclass[ba, noinfoline]{imsart}

\usepackage{cite,amsmath,amssymb,amsfonts,amsthm,url,graphicx,overpic}
\usepackage{caption}
\usepackage{subcaption}
\usepackage[colorlinks,citecolor=blue,urlcolor=blue,filecolor=blue,backref=page]{hyperref}
\usepackage[linesnumbered,ruled]{algorithm2e}

\startlocaldefs
\numberwithin{equation}{section}
\theoremstyle{plain}

\endlocaldefs

\newtheorem{theorem}{Theorem}[section]
\newtheorem{lemma}[theorem]{Lemma}
\newtheorem{definition}[theorem]{Definition}
\newtheorem{corollary}[theorem]{Corollary}

\newcommand{\half}{\frac{1}{2}}
\newcommand{\MD}{\mathcal{MD}_+}
\newcommand{\decisionSpace}{\mathcal{D}}

\newcommand{\bitm}{\begin{itemize}}
\newcommand{\eitm}{\end{itemize}}

\newcommand{\benum}{\begin{enumerate}}
\newcommand{\eenum}{\end{enumerate}}
\newcommand{\kldist}[2]{D\! \left( #1\|#2 \right)}
\newcommand{\reals}{\mathbb{R}}
\newcommand{\parenth}[1] {\left(#1\right)}

\newcommand{\brackets}[1] {\left[#1\right]}

\newcommand{\beqa}{\begin{eqnarray}}
\newcommand{\eeqa}{\end{eqnarray}}
\newcommand{\beqas}{\begin{eqnarray*}}
\newcommand{\eeqas}{\end{eqnarray*}}

\newcommand{\E}{\mathbb{E}}

\def\argmin{\mathop{\arg\!\min}\limits}%
\def\argmax{\mathop{\arg\!\max}\limits}%

\newcommand{\alphabet}[1] { {\mathsf #1}}
\newcommand{\probSimplex}[1]{ \alphabet{P}\parenth{\alphabet{#1}}}

\newcommand{\cX}{\alphabet{X}}

\newcommand{\jac}{J}

\newcommand{\basisfn}[1]{\phi^{(#1)}}

\newcommand{\tp}{\tilde{p}}

\newcommand{\vect}[1]{\boldsymbol{#1}}
\graphicspath{{./figures/}}

\begin{document}

\begin{frontmatter}
\title{Bayesian Lasso Posterior Sampling via Parallelized Measure Transport}
\runtitle{Bayesian Lasso via Parallelized Measure Transport}
\runauthor{M. Mendoza et al.}

\begin{aug}
\author{\fnms{Marcela} \snm{Mendoza}\thanksref{addr1}\ead[label=e1]{mpmendoz@eng.ucsd.edu}},
\author{\fnms{Alexis} \snm{Allegra}\thanksref{addr2}\ead[label=e2]{aallegra@ucsd.edu}}
\and
\author{\fnms{Todd P.} \snm{Coleman}\thanksref{addr1}
\ead[label=e3]{tpcoleman@ucsd.edu}
\ead[label=u1,url]{http://coleman.ucsd.edu/}}

\address[addr1]{Department of Bioengineering, University of California-San Diego
}

\address[addr2]{Department of Electrical Engineering, University of California-San Diego
}


\end{aug}

\begin{abstract}
It is well known that the Lasso can be interpreted as a Bayesian posterior mode estimate with a Laplacian prior. Obtaining samples from the full posterior distribution, the Bayesian Lasso, confers major advantages in performance as compared to having only the Lasso point estimate.  Traditionally, the Bayesian Lasso is implemented via Gibbs sampling methods which suffer from lack of scalability, unknown convergence rates, and generation of samples that are necessarily correlated. We provide a measure transport approach to generate i.i.d samples from the posterior by constructing a transport map that transforms a sample from the Laplacian prior into a sample from the posterior. We show how the construction of this transport map can be parallelized into modules that iteratively solve Lasso problems and perform closed-form linear algebra updates. With this posterior sampling method, we perform maximum likelihood estimation of the Lasso regularization parameter via the EM algorithm. We provide comparisons to traditional Gibbs samplers using the diabetes dataset of Efron et al. Lastly, we give an example implementation on a computing system that leverages parallelization, a graphics processing unit,  whose execution time has much less dependence on dimension as compared to a standard implementation.  

\end{abstract}


\end{frontmatter}

\section{Introduction} \label{sec:intro}



 A quintessential formulation for sparse approximation is Tibshirani's Lasso, which simultaneously induces shrinkage and sparsity in the estimation of regression coefficients \cite{tibshirani1996regression}. The formulation of the standard Lasso is as follows: 
\begin{equation} \label{eqn:l2l1}
x^*= \argmin_{x \in \mathbb{R}^d} ||y-  \Phi x ||_2^2 + \lambda||x||_1 
\end{equation}  
where $y \in \mathbb{R}^n$ is a vector of responses, $\Phi$ is a $n \times d$ matrix of standardized regressors, and $x \in \mathbb{R}^d$ is the vector of regressor coefficients to be estimated.  


    It is known that the Lasso can be interpreted as a Bayesian posterior mode estimate with a Laplacian prior \cite{tibshirani1996regression}.  Imposing a Laplacian prior is equivalent to $L_1$-regularization, which has desirable properties, including robustness and logarithmic sample complexity \cite{ng2004feature}. Various algorithms for solving \eqref{eqn:l2l1} are typically employed, including Least Angle Regression (LARS), iterative soft-thresholding and its successors, and iteratively reweighted least squares(IRLS)  \cite{efron2004least}, \cite{beck2009fast},\cite{friedman2007pathwise}, \cite{friedman2010regularization}. These methods are scalable, yet they only provide a point (maximum a posteriori) estimate.  With i.i.d. samples  $(Z_i: i \geq 1)$ from the posterior distribution, the Bayes optimal decision, $d^*(y)$ can be approximately computed
for any set of possible decisions $\decisionSpace$ and any loss function $l: \reals^d \times \mathcal{D} \rightarrow \reals$ by minimizing the empirical conditional expectation:
\begin{equation} 
d^*(y)=\argmin_{d \in \decisionSpace} \E[l(X,d)|Y=y] \simeq \argmin_{d \in \decisionSpace} \frac{1}{N}\sum_{i=1}^K l(Z_i,d) \label{eqn:empiricalRiskMinimization}
\end{equation}

 Previous approaches have been developed \cite{park2008bayesian,hans2009bayesian} to sample from the posterior distribution corresponding to the Lasso problem based on Markov Chain Monte Carlo methods (MCMC).  However these methods necessarily introduce correlations between the generated samples, are sequential in nature, and do not often scale well with dataset size or model complexity \cite{hastings1970monte},\cite{metropolis1953equation}, \cite{lee2010utility}.


 We here consider a measure transport approach, where i.i.d. samples $(Z_i: i \geq 1)$ from the posterior distribution $P_{X|Y=y}$ associated with the Lasso are produced by first generating i.i.d. samples $(X_i: i \geq 1)$ from the Laplacian prior  $P_X$ and constructing a map $S_y$  that transforms a sample $X$ from $P_X$ into a sample $Z$ from $P_{X|Y=y}$.  From here, the posterior samples are constructed as $(Z_i = S_y(X_i): i \geq 1)$.  We exploit previous results that cast Bayesian posterior sampling from a measure transport perspective \cite{el2012bayesian}, \cite{Marzouk2016},
 that turn the construction of $S_y$ into a relative entropy minimization problem
 \cite{kim2013efficient}, and that exploit the Bayesian Lasso posterior’s log-concavity to turn construction of $S_y$ with polynomial chaos into a convex optimization problem \cite{kim2014dynamic} that is amenable to parallelization \cite{mesa2015scalable}.  In this Lasso setting, we further show that constructing the optimal  map to transform  prior samples to posterior samples can be performed with off-the-shelf Lasso solvers and closed-form linear algebra updates.


\subsection{Relevant Work}
Park and Casella proposed a Gibbs sampler for a variation of the original Bayesian Lasso problem, where the latent variance scale variable in the Gauss-scale mixture representation of the Laplacian distribution has a prior distribution \cite{park2008bayesian}.  This structure leads to a tractable three-step Gibbs sampler that can be used to draw approximate samples from the posterior and construct credibility intervals.  Hans \cite{hans2009bayesian} obviated the need for hyper-parameters and used a direct characterization of the posterior distribution to develop a Gibbs sampler to generate posterior samples.  As an MCMC algorithm, the Gibbs sampler generates a Markov chain of samples, each of which is correlated with its previous sample. The correlation between these samples can decay slowly and lead to burn-in periods where samples have to be discarded \cite{robert2013monte}.  Although theoretical upper bounds on the convergence of Gibbs samplers have been proved \cite{rajaratnam2015mcmc}, these guarantees are weaker in the case of Bayesian Lasso.  \cite{rajaratnam2015fast} developed a two-step Gibbs sampler for the Bayesian Lasso with improved convergence behavior.  However, a way to derive i.i.d. samples from the Bayesian Lasso posterior without burn-in periods has remained elusive.  

Moshely et al. first proposed an alternative method for directly sampling from the posterior distribution based on a measure transport approach \cite{el2012bayesian}, \cite{Marzouk2016} 
using a polynomial chaos expansion \cite{ghanem2003stochastic}. 
Bayesian inference can be cast as a special case of this, where the original distribution is the prior (which in many cases is easy to sample from) and the target distribution is the posterior.  Recently, Kim et al. further investigated the Bayesian transport sampling problem and showed that when the  prior and likelihood satisfy a log-concavity property, the relative entropy minimization approach to find a transport map with a polynomial chaos representation is a convex optimization problem \cite{kim2013efficient}. Mesa et al. introduced an Alternating Direction Method of Multipliers (ADMM) reformulation and showed that this minimization can be performed by iteratively solving a series of convex optimization problems in parallel \cite{mesa2015scalable}. Wang et al. used a measure transport approach to extend the randomize-then-optimize MCMC approach to sample from posteriors with $L1$ priors by transforming the $L1$ prior distribution to a Gaussian distribution \cite{wang2016bayesian}. 



\subsection{Our Contribution}
We present a technique to sample from the Bayesian Lasso posterior based on a measure transport approach \cite{el2012bayesian}, \cite{kim2013efficient}. The formulation is conceptually different from the Gibbs sampler methodology; our initial objective is not to compute samples from the posterior, but rather compute a transport map based upon training samples.  Once the transport map is computed, one can generate an arbitrary number of posterior samples by first computing samples from the prior and passing them through the transport map. We show that construction of the transport map can be performed in a parallelized fashion based on an Alternating Direction Method of Multipliers (ADMM) formulation, as in \cite{mesa2015scalable}.  Unique to the Lasso formulation, our solution only requires off-the-shelf Lasso solvers and linear algebra updates.  This provides opportunities to leverage computing architectures with parallelization and attain significant completion time improvements. 

  We exploit the ability to draw i.i.d. samples from the posterior to develop an Expectation Maximization (EM) algorithm for maximum likelihood estimation of $\lambda$ in \eqref{eqn:l2l1}. Additionally, we compare our results to a Bayesian Lasso Gibbs sampler from \cite{park2008bayesian} and show that we achieve similar results when analyzing the diabetes dataset presented in \cite{efron2004least}. Finally, we show that our  framework is amenable to implementation in architectures that leverage parallelization and provide performance of an example implementation with a graphics pocessing unit (GPU).  

The rest of the paper is organized as follows. In Section~\ref{sec:defns}, we provide some preliminaries and definitions. In Section~\ref{sec:theory}, we introduce a relative entropy minimization formulation for Bayesian Lasso posterior sampling via measure transport.  We consider transport maps described in terms of polynomial chaos and show how under a log-concavity assumption (which applies for the Lasso problem), the relative entropy minimization can be performed with ADMM methods from convex optimization.  We then show how, unique to the Bayesian Lasso, this formulation can be reduced to iteratively solving a collection of Lasso problems in parallel and performing closed-form linear algebra updates. In Section~\ref{sec:EM}, we exploit the ability to draw i.i.d. samples from the posterior to develop an Expectation Maximization (EM) algorithm for maximum likelihood estimation of $\lambda$ in \eqref{eqn:l2l1}. In Section~\ref{sec:comparison}, we compare our results to a Bayesian Lasso Gibbs sampler from \cite{park2008bayesian} and achieve similar results when analyzing the diabetes dataset from Efron et al in \cite{efron2004least}.  In Section~\ref{sec:applications}, we show that our  framework is amenable to implementation in architectures that leverage parallelization and provide performance of an example implementation with a graphics pocessing unit (GPU).  In Section~\ref{sec:discussion-conclusion}, we conclude and discuss future potential directions.

\section{Definitions} \label{sec:defns}
\subsection{Bayesian Lasso Statistical Model}
We consider the following generative model of how a latent and sparse random vector $X \in \reals^d$ relates to a measurement $Y \in \reals^n$:
\begin{equation}
Y= \Phi X + \epsilon  \label{eqn:measurementModel}
\end{equation}
and the measurement noise satisfies $\epsilon \sim \mathcal{N}(0, \sigma^2 I)$.


We assume an i.i.d. Laplacian statistical model on $X$ with parameter $\tau$.  Therefore, the following Bayesian Lasso regression model is specified as:
\beqa
p(y|x; \sigma ^2) & = & \mathcal{N}(y;\Phi x, \sigma^2 \vect{I}_n )\label{eqn:Likelihood}\\
p(x; \tau)& = &\prod_{i=1}^d \frac{\tau}{2} e^{-\tau |x_i|} \label{eqn:LaplacianPrior}
\eeqa 
where $\mathcal{N}(t;\mu,\Sigma)$ represents the density function, evaluated at $t$, of a multivariate normal random variable with expectation $\mu$ and covariance matrix $\Sigma$. 
We note that the negative log posterior density satisfies
\beqa
 -\log p(x|y; \sigma^2, \tau ) \propto \frac{1}{2 \sigma^2} \|y-\Phi x\|_2^2 + \tau||x||_1  \label{eqn:posteriorBLASSO}
\eeqa
As such, the standard Lasso problem for a given $\lambda \equiv 2\tau\sigma^2$
\begin{equation} \label{eqn:l2l1_2}
x^*= \argmin_{x \in \mathbb{R}^d} ||y-  \Phi x ||_2^2 + \lambda||x||_1 
\end{equation}
is a maximum a posteriori estimation problem for the Laplacian prior in \eqref{eqn:LaplacianPrior}.

In the model above and throughout using our methodology, we assume that the parameter $\sigma^2$ is fixed and known, deviating from the results from the Bayesian Lasso Gibbs sampler first presented in \cite{park2008bayesian}, for which $\sigma^2$ is imparted with a prior. As such, we are considering the posterior distribution associated with the original Bayesian interpretation to Lasso, for which the solution to \eqref{eqn:l2l1} is the MAP estimate.
We will discuss the assumptions associated with \cite{park2008bayesian} in Section~\ref{sec:comparison}.





\subsection{Transport Maps}
\renewcommand{\cX}{\reals^d}

Define $\probSimplex{\cX}$ as the space of probability measures over $\cX$ endowed with the Borel sigma-field.

\begin{definition}[Push-forward]
Given $P \in \probSimplex{\cX}$ and $Q \in \probSimplex{\cX}$, we
say that the map $S: \cX \to \cX$ \textbf{pushes forward} $P$ to
$Q$ (denoted as $S \# P = Q$) if a random variable $X$ with
distribution $P$ results in $Z \triangleq S(X)$ having
distribution $Q$.
\end{definition}
This gives insight into the notion of ``transport maps'' because they transform random variables from one distribution to another.

We say that $S: \cX \to \cX$ is a {\it diffeomorphism} on $\cX$ if
$S$ is invertible and both $S$ and $S^{-1}$ are differentiable.
 We say that a diffeomorphism $S$ is {\it monotonic} on $\cX$ if
its Jacobian $J_S: \cX \to \reals^{d \times d}$ satisfies the property that $J_s(x)$ is positive definite for all $x \in \cX$, e.g.
\beqa
J_S(x) \succ 0 \quad \forall\; x \in \cX
\eeqa
Define the set of all monotonic diffeomorphisms as $\MD$.

With this, we have the following lemma from standard probability:
\begin{lemma}\label{lemma:SpushPtoQ:equiv:JacobianEqn}
Consider any $S \in \MD$ and $P$,
$Q\in \probSimplex{\cX}$ that both have the densities $p$, $q$ with
respect to the Lebesgue measure. Then $S \# P = Q$ if and only
if
\begin{eqnarray}
p(x) = q(S(x)) \det \parenth{J_S(x)} \quad \forall x \in \cX
\label{eqn:defn:JacobianEqn}
\end{eqnarray}
where $J_{S}(x)$ is the Jacobian matrix of the map $S$ at $x$.
\end{lemma}

We note from a classical result in optimal transport theory \cite{brenier1991polar,villani2008optimal} that if $P$ and $Q$ have densities $p$ and $q$ with respect to the Lebesgue measure, there will always exists a transport map $S$ that is a monotonic diffeomorphism (e.g. $S \in MD$)
that pushes $P$ to $Q$.

Note that in the case of Bayesian Lasso, $P\equiv P_X$ is the prior distribution on $X$, which has a Laplacian density $p$ given by
\eqref{eqn:LaplacianPrior}, and $Q\equiv P_{X|Y=y}$ is the posterior distribution on $X$, which has a density described up to a proportionality constant by \eqref{eqn:posteriorBLASSO}.

\section{Bayesian Lasso via Measure Transport} \label{sec:theory}

In this section, we provide background on measure transport theory and show that for the Bayesian LASSO,  we can efficiently find a transport map that transforms samples from the prior distribution in \eqref{eqn:LaplacianPrior} to samples from the posterior.  We utilize the ADMM framework introduced in \cite{mesa2015scalable} and develop a Bayesian Lasso transport map by iteratively solving a collection of Lasso problems (which themselves can be solved with existing efficient sparse approximation algorithms) in parallel and performing linear algebra updates.  

As an alternative to the Gibbs sampling approaches described previously, we consider finding a transport map $S$ that pushes the prior $P=P_X$ to the posterior $Q=P_{X|Y=y}$.  Throughout, we will assume $p$ and $q$ are the densities associated with the prior and posterior, respectively.
Given such a map, we can sample from the posterior distribution by transforming i.i.d. samples $(X_i: i \geq 1)$ with $S_y$  to i.i.d. samples $(Z_i = S_y(X_i): i \geq 1)$ from the posterior.  Figure \ref{fig:priortoposterior} shows the effect of a transport map on samples for the Bayesian Lasso in \eqref{eqn:LaplacianPrior}.

\begin{figure*}
\centering
\begin{subfigure}{.5\textwidth}
  \centering
  \includegraphics[width=\linewidth]{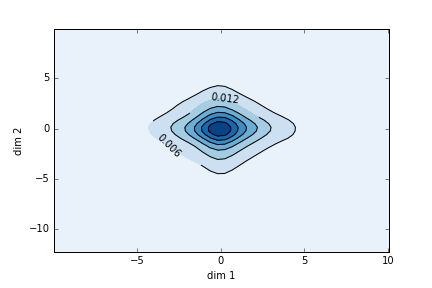}
  \caption{ }
  \label{fig:prior}
\end{subfigure}%
\begin{subfigure}{.5\textwidth}
  \centering
  \includegraphics[width=\linewidth]{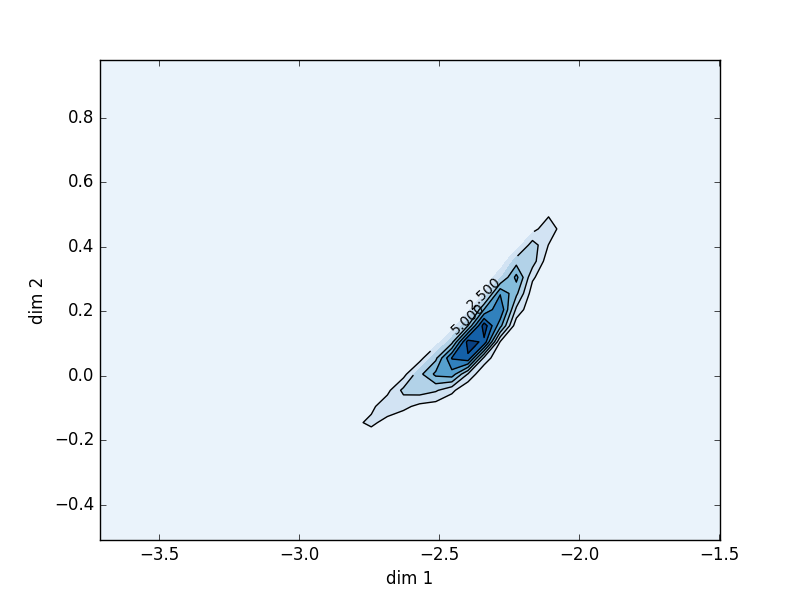}
  \caption{ }
  \label{fig:posterior}
\end{subfigure}

\caption{$\textbf{ Effect of transport map $S$ on prior samples}$  (a) kernel density estimate of prior (Laplacian) distribution constructed by samples (b) kernel density estimate of samples transformed through transport map $S$; posterior density }
\label{fig:priortoposterior}
\end{figure*}

\subsection{A Convex Optimization Formulation}

Recent work \cite{kim2013efficient,kim2014dynamic} has shown that for the problems where the prior density $p(x)$ and likelihood $p(y|x)$ are log-concave in $x$, constructing a transport map that pushes $P_X$ to $P_{X|Y=y}$ can be performed with convex optimization.

Consider fixing $q$ as the true posterior density. Given an arbitrary $S \in \MD$, then there will be an induced $\tilde{P}_S$ (with density $\tp_S$( for which $S$ pushes $\tilde{P}_S$ to $Q$. That is, from the Jacobian equation: 
\begin{equation}
 \tp_S(u) = q(S(u)) \det(J_S(u)) \quad \text{ for all } u \in \reals^d \label{eqn:defn:arbJacobianEqn2}
\end{equation}
For an arbitrary $S$, $\tilde{P}_S$ need not be the same as $P$.  From this perspective, we re-define our problem as finding the transport map $S^*$ that minimizes a distance between $P$ and the induced $\tilde{P}_S$.  We use the Kullback-Leibler divergence and arrive at the following optimization problem: 
\begin{equation}
   S^* = \argmin_{S \in \MD} \kldist{P}{\tilde{P}_S} \label{eq:minKL} 
\end{equation}

By defining 
\begin{equation}
g(x) \triangleq -\log   f_{Y|X}(y|x) -  \log f_X(x) 
\end{equation}
where the densities refer to the likelihood and prior, \eqref{eq:minKL} becomes

\begin{equation}
 S^* = \argmin_{S \in \MD} \E_P \brackets{ g(S(X)) - \log\det \parenth{J_{S(X)}} } \label{eq:gen-push-for-obj}
\end{equation}

Moreover, when $q$ is log-concave (equivalently when $g$ is convex), this (infinite-dimensional) optimization problem is convex.    Moreover, note that the proportionality constant in the denominator of the posterior density for Bayes rule is not required to perform this minimization, because it does not vary with $x$.

\noindent \textbf{Parametrization of the Transport Map:}
 In order to solve \eqref{eq:gen-push-for-obj}, we parametrize the problem to arrive at a finite-dimensional convex optimization problem. We approximate any $S \in \MD$ as a linear
combination of basis functions through a Polynomial Chaos Expansion (PCE) \cite{xiu2002wiener}, \cite{ernst2012convergence} where $\phi$ are the polynomials orthogonal with respect to the prior $P$:
\beqa
  &&S(x) = \sum_{j \in \mathcal{J}}b_j \basisfn{j}(x)
  \label{eqn:defn:linearbasis}\\
  &&\int_{x \in \cX} \basisfn{i}(x)\basisfn{j}(x)  p(x) dx = \delta _{i,j} \label{eq:PCE}
\eeqa
with $\delta_{i,j}$ being $1$ if $i=j$ and $0$ otherwise. 
Now define  $K=|\mathcal{J}|$ and we have that: \beqa
  F         &=& [b_1,\ldots, b_K],  \qquad\qquad\qquad\;\;\; d \times K \label{eqn:polynomialchoas:a}\\
  A(x)      &=& [ \phi^{(1)}(x),\ldots, \phi^{(K)}(x)]^T, \qquad K \times 1 \label{eqn:polynomialchoas:b}\\
  S(x)      &=& FA(x) , \qquad\qquad\qquad\qquad\;\;\; d \times 1 \label{eqn:polynomialchoas:c} \\
  \jac(x) &=& \brackets{\frac{\partial  \phi^{(i)}}{\partial x_j}(x)}_{i,j},  \qquad\qquad\quad  K \times d\label{eqn:polynomialchoas:d}\\
  \jac_S(x) &=& F\jac(x)  \qquad\qquad\qquad\qquad\;\; d \times d\label{eqn:polynomialchoas:e}.
\eeqa
We can then approximate the expectation from \eqref{eq:gen-push-for-obj} using an empirical expectation based upon i.i.d. samples $(X_i: i \geq 1)$ from the prior $P$. Within the context of the Bayesian Lasso, $P$ is the Laplacian (doubly exponential) distribution which is easy to sample from. Letting $A_i \triangleq A(X_i) \in \reals^{K \times 1}$ and $J_i \triangleq J(X_i) \in \reals^{K \times d}$, we arrive at the following finite-dimensional problem: 
\beqa
  F^*  = \argmax_{F:FJ_i \succ 0} \frac{1}{N} \sum_{i=1}^N  g(FA_i) -\log\det\parenth{F\jac_i}\label{eq:BayesF}
\eeqa
Whenever $q$ is log-concave (equivalently $g$ is convex), this is a finite-dimensional convex optimization problem.  Moreover, as $K \to \infty$, from the PCE theory, the map $F^*A(x)$ converges (in the relative entropy sense) to a transport map $S^*$ that pushes $P$ to $Q$ (see \cite{kim2015tractable}).

\subsection{Parallelized Convex Solver with ADMM}
More recently, \cite{mesa2015scalable} demonstrated a scalable framework to solve \eqref{eq:BayesF} which only requires
iterative linear algebra updates and solving, in parallel, a number of quadratically regularized point estimation problems. The distributed architecture involves an augmented Lagrangian and a concensus Alternating Direction Method of Multipliers (ADMM) formulation:  
\beqas
\min_{F,Z,p,B} &   \frac{1}{N} \sum_{i=1}\limits^N g(p_i) -\log \det Z_i  + &\half \rho \|F_i -B \|_2^2  \\
                      &+ \frac{1}{N} \sum\limits_{i=1}^N \half \rho \| BA_i - p_i\|_2^2 +  &\half \rho \| BJ_i - Z_i\|_2^2 \\
\text{s.t.}  & BA_i = p_i: \qquad\qquad &\gamma_i \quad (d \times 1) \\
             & BJ_i = Z_i: \qquad\qquad &\beta_i \quad (d \times d) \\
             & F_i - B = 0: \qquad \quad & \alpha_i \quad (d \times K) \\
             & Z_i \succ 0 \qquad \qquad \qquad&
\eeqas
for any fixed $\rho>0$.
\newcommand{\goleftstuff}{\!\!\!}

A penalized Lagrangian is solved iteratively by first solving for $B^{k+1}$
\begin{align}
B^{k+1}=&\frac{1}{N} \sum_{i=1}^{N} \big[ \rho\left(F_i^k + p^{k}_i 
A_i^{T}+Z_i^{k}J_i^{T}\right) \nonumber\\
   	\quad&+\gamma_i^{k}A_i^{T}+\beta_{i}^{k}J_i^{T}+\alpha_{i}^{k} \big]  \mathcal{M}, \label{eqn:ADMM:B} \\
\mathcal{M} \triangleq& \left[ \rho \left ( I + \frac{1}{N}\sum_{i=1}^{N} A_i A_i^{T} +
J_i J_i^{T} \right ) \right]^{-1} \label{eqn:ADMM:inverse} 
\end{align}
and then solving, in parallel for $1 \leq i \leq N$, the other variable updates:
\begin{subequations}
\label{eqn:ADMM}
\begin{align}
 F_i^{k+1} =& -\frac{1}{\rho} \alpha^{k}_i+B^{k+1} \label{eqn:ADMM:F} \\
 Z_i^{k+1} =&Q\tilde{Z}_{i}Q^T  \label{eqn:ADMM:Z} \\
p_i^{k+1} =& \argmin_{p_i} g(p_i) +   \half \rho \| B^{k+1} A_i - p_i\|_2^2 \nonumber\\
&+ \gamma_i^{kT} (p_i-B^{k+1} A_i ) \label{eqn:ADMM:p}\\
\gamma_i^{k+1} =& \gamma_i^k + \rho (p_i^{k+1}-B^{k+1}A_i)  \label{eqn:ADMM:gamma} \\
\beta_i^{k+1} =& \beta_i^k + \rho(Z_i^{k+1} - B^{k+1}J_i)  \label{eqn:ADMM:lambda} \\
\alpha_i^{k+1} =& \alpha_i^k + \rho(F_i^{k+1}-B^{k+1})\label{eqn:ADMM:alpha}
\end{align}
\end{subequations}
ADMM guarantees convergence to the optimal solution \cite{boyd2011distributed}. 
To emphasize, each $i$th update in \eqref{eqn:ADMM} can be solved in parallel. As  \eqref{eqn:ADMM:Z} is an eigenvalue-eigenvector decomposition (details can be found in \cite{mesa2015scalable}), it follows that all the updates involve linear algebra  with the exception of \eqref{eqn:ADMM:p}, which is a quadratically regularized point estimation problem.

\subsection{Efficiently Solving the Bayesian Lasso}
 We exploit the unique problem structure of Bayesian Lasso to simplify a scalable implementation. 

\begin{lemma}
The PCE for the Laplacian distribution is $\phi_L(x)=\phi_E(|x|)$ where $\phi_E$ are the Laguerre polynomials.
\end{lemma}

\begin{proof}
\begin{align}
\int\limits_{- \infty}^{\infty} \phi_E^i(|x|) \phi_E^j(|x|) p_L(x) dx &=
\int\limits_{- \infty}^{\infty} \phi_E^i(|x|) \phi_E^j(|x|) \frac{1}{2} p_E(|x|) dx \nonumber 
\\
&= 2 \int\limits_{0}^{\infty} \phi_E^i(x) \phi_E^j(x) \frac{1}{2} p_E(x) dx \nonumber \\
&= \delta _{i,j}
\end{align}
Where the first equality holds because the Laplacian density $p_L(x)$ is related to the exponential density $p_E(x)$ by  $p_L(x) = \frac{1}{2} p_E(|x|)$, the second equality holds by symmetry of the function being integrated, and the third follows because the PCE for the exponential distribution is obtained with the Laguerre polynomials $\phi_E^{(j)}$ \cite{xiu2002wiener}.
\end{proof}

We now show that for Bayesian Lasso, the only ADMM update that is not linear algebra is simply a Lasso problem.  
\begin{theorem}
For the Bayesian Lasso statistical model given by \eqref{eqn:posteriorBLASSO}, the ADMM update \eqref{eqn:ADMM:p} is a d-dimensional Lasso point estimation problem:
\begin{eqnarray}
p_i^{k+1} &=& \argmin_{p_i}  ||\hat{y}- \hat{\Phi}^T p_i||^2_2 + \lambda ||p_i||_1 \label{eqn:lassoADMM}
\end{eqnarray} 
where $\hat{\Phi}$ and $\hat{y}$ satisfy
\begin{eqnarray}
\hat{\Phi} ^T\hat{\Phi}&=& \Phi^T \Phi + \frac{1}{2} \rho I \label{eqn:defn:theorem:BayesLasso:a}\\
\hat{y} &=& \parenth{ \brackets{y^T\Phi + \half \rho (B^{k+1}A_i)^T -\frac{1}{2}\gamma_i^{kT}} \hat{\Phi}^+}^T \nonumber 
\end{eqnarray}
and $\hat{\Phi}^+$ represents the pseudo-inverse.
\label{thm:BayesianLASSO:ADMM}
\end{theorem}

\newcommand{\quadratic}{\text{quad}}

\begin{proof}
Dropping (i) superscript indices of \eqref{eqn:ADMM:p}, we have 
\begin{align}
p^*  
	 &=\argmin_p \;\; \quadratic(p) + \lambda ||p||_1, \nonumber\\
\quadratic(p) &\triangleq p^T(\Phi^T \Phi + \frac{1}{2}\rho I )p +(\gamma^T-2y^T\Phi - \rho (BA)^T)p. \nonumber \\
              &=p^T \hat{\Phi}^T \hat{\Phi} p +(\gamma^T-2y^T\Phi - \rho (BA)^T)p \label{eqn:proof:theoremBayesLasso:b}
\end{align}
where \eqref{eqn:proof:theoremBayesLasso:b} follows from
performing a Cholesky decomposition to build a unique $\tilde{\Phi} \in \reals^{d \times d}$ and then zero padding to build
$\hat{\Phi} \in \mathbb{R}^{n \times d}$, obeying the relationship given in \eqref{eqn:defn:theorem:BayesLasso:a}.  Then we complete the square in order to get an equation of the form $\| \hat{\Phi} p\|_2^2 - 2 \hat{y}^T \hat{\Phi} p + \|\hat{y} \|_2^2 = \|\hat{y} - \hat{\Phi} p \|_2^2$:
\begin{eqnarray*}
-2\hat{y}^T\hat{\Phi}p = (\gamma ^T-2y^T \Phi - \rho (BA)^T  ) p.
\end{eqnarray*}
\end{proof}

\begin{corollary}
For the Bayesian Lasso, the problem of finding a map $S^*$ to generate i.i.d. samples from $P_{X|Y=y}$ solved by iteratively solving for linear algebra updates and solving, in parallel, a collection of $d$-dimensional Lasso problems \eqref{eqn:l2l1}.   
\end{corollary}

The procedure for Bayesian Lasso via measure transport is outlined in Algorithm \ref{blassoAlg}. 

\begin{algorithm} 
\caption{Parallelized Bayesian Lasso}\label{blassoAlg}
\SetKwInOut{Input}{Input}
\SetKwInOut{Output}{Output}
\underline{function BayesianLasso} $(\textbf{x}_1,...,\textbf{x}_N \in \mathbb{R}^d$, $\textbf{y} \in \mathbb{R}^n$, $\Phi \in \mathbb{R}^{n \times d}$, $\lambda$, $\rho$, $K$)\;
\Input{Samples $\textbf{x}_1,..., \textbf{x}_N$ from prior in \eqref{eqn:LaplacianPrior}}  
\Output{$B^{\infty}$ holds coefficients of map $S$ such that $S(x)=B^{\infty}A(x)$}
Construct $A_i$ and $J_i$ via Polynomial Chaos Expansion for $i=1,...,N$ as in \eqref{eqn:polynomialchoas:b} and \eqref{eqn:polynomialchoas:d}\;
Construct $\mathcal{M}$ as in \eqref{eqn:ADMM:inverse}\;
Initialize $B^0$ and $F_i^0, Z_i^0, p_i^0, \gamma_i^0, \beta_i^0, \alpha_i^0$ randomly for $i=1,...,N$\;  
\While{$B^k$ has not converged}{
  Update $B^{k+1}$ as in \eqref{eqn:ADMM:B} \;
  Update in parallel for $i=1,...,N$ $F_i^{k+1}, Z_i^{k+1}, \gamma_i^{k+1}, \beta_i^{k+1}, \alpha_i^{k+1}$ as in \eqref{eqn:ADMM} \\
  \qquad and $p_i^{k+1}$ with a $\textbf{Lasso solver}$ as in \eqref{eqn:lassoADMM} \;
  $k=k+1$ 
 }
 
\end{algorithm}

\section{Choosing $\lambda$ via Maximum Likelihood Estimation} \label{sec:EM}

The parameter of the standard Lasso in \eqref{eqn:l2l1}, $\lambda$,  can be chosen by cross-validation, generalized cross-validation, and ideas based on unbiased risk minimization \cite{tibshirani1996regression}. Park and Casella used empirical Bayes Gibbs sampling \cite{casella2001empirical} to find marginal maximum likelihood estimates of hyperparameters via EM algorithm \cite{park2008bayesian}.
This empirical scheme, however, is specific to the Gibbs sampler and the hierarchical model introduced in \cite{casella2001empirical}. Here, we propose an EM algorithm to calculate a maximum likelihood estimate of $\lambda$ according to the statistical model where $\tau$ and $\sigma^2$ (and thus $\lambda=2\tau \sigma^2$) are assumed fixed and non-random.  Without loss of generality, for the remainder of this section, we assume here that $\sigma^2 = \frac{1}{2}$ and so $\lambda=2\tau \sigma^2 \equiv \tau$.

As such, our statistical model is of the form $p(x,y;\tau) \equiv p(x,y;\lambda)$ where $p(x;\lambda)$ is a Laplacian density with parameter $\lambda$ and $p(y|x)=\mathcal{N}(y;\Phi x, \frac{1}{2} I_n)$.

In the EM framework, it is our objective to find
\begin{eqnarray}
\hat{\lambda} &=& \argmax_\lambda p(y;\lambda) \\
              &=& \argmax_\lambda \int  p(x,y;\lambda) dx
\end{eqnarray}
This can be performed by iteratively solving in the E-step for the posterior distribution with $\lambda^{(k)}$:
$g_k(x) \triangleq p(x|y;\lambda^{k})$.
Since we can generate i.i.d. posterior samples with our transport map $S^*$, define $(Z_i^{(k)}: i \geq 1)$ as i.i.d. samples from $p(x|y;\lambda^{k})$.
The M-step is found by solving for $\lambda^{(k+1)}$:
\begin{eqnarray}
 \lambda^{(k+1)} &=& \argmax_\lambda \mathbb{E}_{g_k}\left[\log p(X,y;\lambda)|Y=y\right]  \label{eqn:EMalg:M-step:general} \\
 &=& \argmax_\lambda d \log \lambda - \mathbb{E}_{g_k}\left[ \|X\|_1 |Y=y\right] 
 \label{eqn:EMalg:M-step:BayesLasso:a}\\
 &=& \frac{d}{\mathbb{E}_{g_k}\left[ \|X\|_1 |Y=y\right]}\label{eqn:EMalg:M-step:BayesLasso:b} \\
 &\simeq&  \frac{d}{\frac{1}{N} \sum_{i=1}^N \|Z_i^{(k)}\|_1 }  \label{eqn:EMalg:M-step:BayesLasso:c}
 \end{eqnarray}
 where \eqref{eqn:EMalg:M-step:BayesLasso:a} follows from \eqref{eqn:Likelihood} and \eqref{eqn:LaplacianPrior} with $\sigma^2=\frac{1}{2}$ and $\tau=\lambda$, 
 \eqref{eqn:EMalg:M-step:BayesLasso:b} follows from closed form minimization, and \eqref{eqn:EMalg:M-step:BayesLasso:c} follows from approximating the conditional expectation with an empirical expectation using i.i.d. $(Z_i^{(k)}: i \geq 1)$ from $g_k$.


Altogether, this becomes

\begin{enumerate}
\item Choose an initial $\lambda ^{(1)}$ and set $k=1$
\item (E-Step): Perform Algorithm 1  with $\lambda= \lambda^{(k)}$ to find $S$ and generate $N$ samples from the posterior distribution as $Z_j^{(k)}= S(X_j)$ for $j=1...N$  where $(X_i: i \geq 1)$ are i.i.d. samples from $p(x;\lambda^{(k)})$ in \eqref{eqn:LaplacianPrior}. 
\item (M-Step):  Update $\lambda^{k+1}$ according to \eqref{eqn:EMalg:M-step:BayesLasso:c} 
\item If convergence has occurred, stop; else let $k=k+1$ and return to step 2. 
\end{enumerate}



\section{Comparisons to Gibbs Sampling} \label{sec:comparison}

We now present comparisons of our measure transport methodology with a Gibbs sampler for Bayesian Lasso based on the diabetes data of Efron et al \cite{efron2004least}, which has $d=10$.  We will follow the analysis presented in \cite{park2008bayesian} and compare our results with the respective Gibbs sampler.  We note that our Bayesian Lasso model differs from that in \cite{park2008bayesian} in that we do not place a prior on $\sigma^2$. We first introduce and clarify models of Bayesian Lasso.  Then we compare regression estimates obtained using Gibbs sampling with those obtained using our methodology. 


\subsection{Models of Bayesian Lasso}

Park and Casella \cite{park2008bayesian} extend the Bayesian Lasso regression model to account for uncertainty in the hyperparameters by not only placing prior on $\tau$ in accordance with the interpretation of the Laplacian as a Gauss Scale mixture, but in addition placing a prior on $\sigma^2$.  Haans et al. \cite{hans2010model} introduced computational approaches to handle model uncertainty under the Bayesian Lasso regression model and provided an implementation to run Gibbs sampling for different regression models.  Here we give a summary of these varying Bayesian Lasso regression models. 

\subsubsection{Succinct Fixed Parameters}

A ``succinct'' fixed parameter model operates with $\tau$ (proportional to $\lambda$) and $\sigma ^2$ fixed and non-random. We denote the probability space induced by this model as $ (\Omega, \mathcal{F}, \mathbb{P})$. 
The generative model is expressed by \eqref{eqn:Likelihood} and \eqref{eqn:LaplacianPrior}.  Therefore, for any $\omega =(x,y) \in \Omega$, the posterior distribution $P_{X|Y=y}$ can be expressed in density form as
\begin{eqnarray}
p(x | y; \sigma^2, \tau) \propto p(y|x; \sigma^2) p(x; \tau) \\
\propto \exp \left( -\frac{1}{2\sigma^2} \|y-\Phi x\|_2^2 -\tau \|x\|_1 \right)
\end{eqnarray}
 In our Measure Transport approach, we operate on a fixed parameter model. We also refer to this model as the Bayesian-frequentist approach.  

\subsubsection{Fixed Parameter Gauss Scale Mixtures Model}
In order to develop Gibbs samplers, it is natural to interpret the Laplacian distribution as a scale mixture of Gaussians.


We denote the probability space induced when  $\tau$ and $\sigma^2$ are still non-random and 
the prior \eqref{eqn:LaplacianPrior} is  represented as a Gauss-scale mixture with latent random variables $(t_1^2,\ldots,t_d^2)$ as $(\tilde{\Omega}, \tilde{\mathcal{F}}, \tilde{\mathbb{P}})$ . For any $\tilde{\omega}=(x,y, t_1^2,...,t_2^2) \in \tilde{\Omega}$, the posterior distribution on $X$ given $Y$ and latent variables
\begin{equation}
p(x|y, t_1^2,...,t_d^2; \sigma^2, \tau)
\end{equation}
is proportional to the product of the following probabilities 
\begin{eqnarray}
p(\vect{y}| \vect{x}; \sigma^2, \Phi) &=& \mathcal{N}_n(  \Phi \vect{x}, \sigma^2 \vect{I}_n ) \\
p(\vect{x}| t_1^2,..., t_d^2)  &=& \mathcal{N}_d(\vect{0}, D_t) \quad D_t=\text{diag}(t_1^2,...,t_d^2)\\
p(t_1^2,...,t_d^2; \tau) &=& \prod_{j=1}^d \frac{\tau^2}{2} e^{-\tau^2 t_j^2/2}dt_j^2
\end{eqnarray}

Note that for any $A \in \tilde{\mathcal{F}}$ for which $A \in \mathcal{F}$, for instance $A=\left\{\|X\|_2^2 > 3, Y \in [2,2.01]\right\}$, it follows that $\mathbb{P}(A)=\tilde{ \mathbb{P}}(A)$.  In other words, $\mathbb{P}$ is the restriction of $\tilde{\mathbb{P}}$ to $\mathcal{F}$: $\mathbb{P}=\tilde{\mathbb{P}}_{|\mathcal{F}}$. Thus this is also a Bayesian-frequentist approach; its primary purpose is for using latent variables for Gibbs sampling.

\subsubsection{Prior on $\sigma^2$ }
Another Bayesian perspective formulated by
Park and Casella \cite{park2008bayesian}
 treats $\sigma^2$ as a random variable with a prior density $\pi(\sigma^2)$, inducing a probability space $(\Omega' , \mathcal{F}', \mathbb{P}')$. 
 The prior on $X$ given $\sigma^2$ is the standard Laplacian
\begin{equation}\label{eqn:parkprior}
p(x|\sigma^2;\tau)= \left(\frac{\tau}{2\sigma}\right)^d
\exp \left(-\frac{\tau}{\sigma} \|x\|_1 \right)
\end{equation}
that is, the penalty parameter is now scaled by the squared root of the error variance.  
This gives rise to the following hierarchical representation of the full model: 
\begin{eqnarray}
p(y| x, \sigma^2; \tau) &=& \mathcal{N}_n(  \Phi x, \sigma^2 I_n ) \\
p(x| t_1^2,..., t_d^2,  \sigma^2)  &=& \mathcal{N}_d(0, \sigma^2 D_t) \quad D_t=\text{diag}(t_1^2,...,t_d^2)\\
p(t_1^2,...,t_d^2; \tau) &=& \prod_{j=1}^d \frac{\tau^2}{2} e^{-\tau^2 t_j^2/2}dt_j^2\\
p(\sigma^2) &=& \pi(\sigma^2)
\end{eqnarray}
 For any $\omega'=(x, y, t_1^2,...,t_d^2, \sigma^2) \in \Omega' $ the posterior distribution on $(x,\sigma^2)$ given y $y$  takes the form: 
\begin{eqnarray}
&p(x, \sigma^2 | y ; \tau)  \propto  \\
&\pi(\sigma^2) (\sigma^2)^{- (n-1)/2} \exp -\frac{1}{2\sigma^2} \|y - \Phi x\|_2^2  -\tau \|x\|_1
\end{eqnarray}

\subsection{Analysis on Diabetes Data}
We analyze a diabetes data set \cite{efron2004least} and compare results when using the Gibbs sampler presented in \cite{park2008bayesian} which utilizes a prior on $\sigma^2$ (operating on $(\Omega' , \mathcal{F}', \mathbb{P}')$) and when using samples from our transport based methodology (operating on  $(\Omega, \mathcal{F}, \mathbb{P})$). We show that despite our treatment of $\sigma^2$ as fixed, we achieve similar results as we capture the complexity of the posterior distribution in this real dataset.  

Figure \ref{fig:lassopaths} compares our measure transport Bayesian Lasso posterior median estimates \ref{fig:sub3} with the ordinary Lasso \ref{fig:sub1} and the Gibbs sampler posterior median estimates \ref{fig:sub2}. We take the vector of posterior medians as the one that minimizes the $L_1$ norm loss averaged over the posterior.  For all three methods, the estimates were computed by sweeping over a grid of values for $\lambda$.  We implemented our measure transport Bayesian Lasso with a Polynomial Chaos Expansion order of 3, and trained with $N=500$ prior samples to compute a transport map.  The specifications for the Gibbs sampler were to use a scale-invariant prior on $\sigma ^2$ and to run for 10,000 iterations after 1000 iterations of burn-in.

Figure \ref{fig:sub3} shows the resulting optimal $\lambda$ (depicted with a vertical line) found by the EM algorithm presented in Section \ref{sec:EM}. The vertical line in Figure \ref{fig:sub2} is the optimal $\lambda$ found by \cite{park2008bayesian} (with prior on $\sigma^2$) by running a Monte Carlo EM algorithm  corresponding to the particular Gibbs implementation. The vertical line in the Lasso graph \ref{fig:sub1} represents the estimate chosen by n-fold cross validation.

Despite treating $\sigma^2$ as fixed, the $L_1$ paths are very similar to the Bayesian Lasso imparted with a prior on $\sigma^2$.  As already noted in previous work, the Bayesian Lasso paths are smoother than the Lasso estimates.

\begin{figure*}
\centering
\begin{subfigure}{.33\textwidth}
  \centering
  \includegraphics[width=\linewidth]{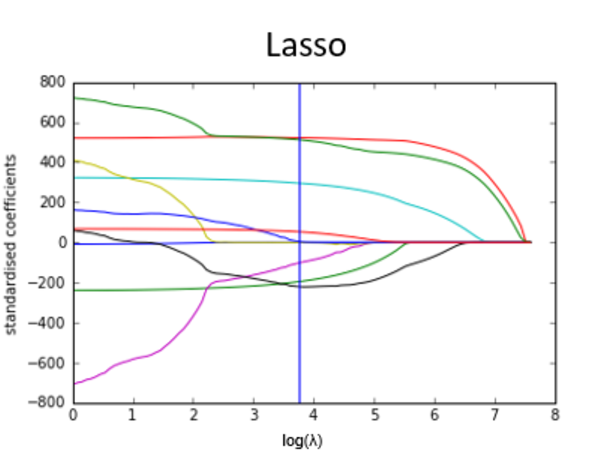}
  \caption{ }
  \label{fig:sub1}
\end{subfigure}%
\begin{subfigure}{.33\textwidth}
  \centering
  \includegraphics[width=\linewidth]{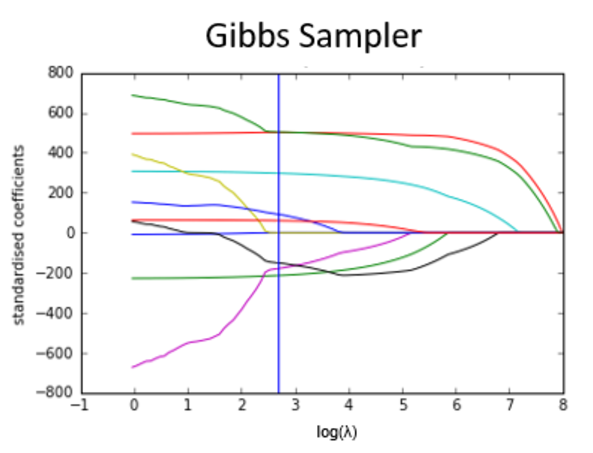}
  \caption{ }
  \label{fig:sub2}
\end{subfigure}
\begin{subfigure}{.33\textwidth}
\centering
  \includegraphics[width=\linewidth]{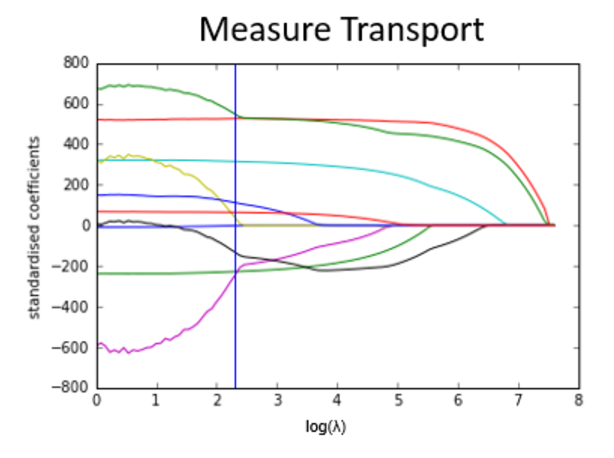}
  \caption{ }
  \label{fig:sub3}
\end{subfigure}
\caption{$\textbf{Comparison of Linear Regression Estimates on Diabetes Data}$ trace plots for estimates of the diabetes data regression parameters for (a) Lasso (b) Gibbs sampler Bayesian Lasso; (c) our measure transport Bayesian Lasso method. The vertical line represents the $\lambda$ estimate found by cross-validation, EM for Gibbs sampling, and EM for measure transport method respectively.  }
\label{fig:lassopaths}
\end{figure*}



   
   We further compare the 95\% credible intervals for the diabetes data obtained with a fixed $\lambda$ (the optimal $\lambda$ corresponding to the Gibbs sampler) for the marginal posterior distributions of the Bayesian Lasso estimates. Figure \ref{fig:confidence_int} shows the corresponding result for the Lasso, Gibbs sampler, and our proposed methodology. The Gibbs sampler credible intervals are wider in comparison due to the treatment of $\sigma^2$ as random.  
   
   Figure \ref{fig:kdes} shows  Kernel  Density  Estimates of  two  of  the regression  variables obtained by Gibbs sampling and sampling with a transport map.  The density estimates obtained through a Gibbs sampler operating on $(\tilde{\Omega}, \tilde{\mathcal{F}}, \tilde{\mathbb{P}})$ are similar  in shape and support to those from a transport map, verifying the convergence to the posterior distribution in both methods.  The small differences seen in the estimates might be due to the polynomial nature of the transport map. We also show kernel density estimates obtained with a Gibbs sampler operating on $(\Omega' , \mathcal{F}', \mathbb{P}')$ for comparison.

\begin{figure}
\centering
\includegraphics[width=\linewidth]{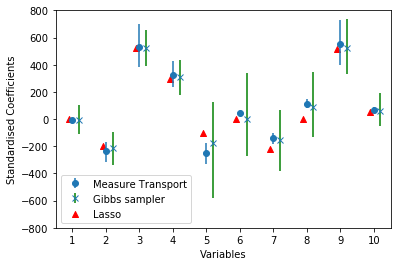}
\caption{Posterior median Bayesian Lasso estimates and corresponding 95 percent credible intervals for a Gibbs sampler and our Measure Transport methodology.  Lasso estimates are also shown for comparison.}
\label{fig:confidence_int}
\end{figure}

\begin{figure*}
\centering
\begin{subfigure}{.5\textwidth}
  \centering
  \includegraphics[width=\linewidth]{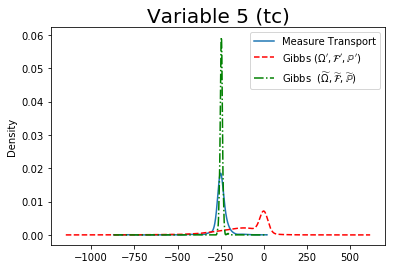}
  \caption{ }
  \label{fig:kdesub1}
\end{subfigure}%
\begin{subfigure}{.5\textwidth}
  \centering
  \includegraphics[width=\linewidth]{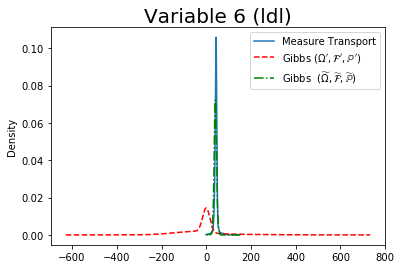}
  \caption{ }
  \label{fig:kdesub2}
\end{subfigure}
\caption{Marginal posterior density estimates for variables 5 and 6 of the Diabetes dataset.  Kernel density estimates were constructed using 10,000 samples from a Gibbs sampler or a transport map respectively. }
\label{fig:kdes}
\end{figure*}

\section{Parallelized Implementation and Applications} \label{sec:applications}

The fundamentally parallel nature of our Bayesian Lasso formulation allows for solution implementation on a variety of platforms. The fact that our solution relies solely on linear algebra and Lasso solvers allows for it to be deployed in a variety of architectures for parallel computing.  
In order to leverage the parallel nature of the algorithm presented above, we here present implementation with a Iterative-Reweighted Least Squares (IRLS) Lasso solver implemented in a Graphics Processing Unit (GPU) solution. 



\subsection{IRLS solver within a GPU Implementation}
  In the last several years, GPUs have gained significant attention for their parallel programmability. In this work, we made use of the ArrayFire library that abstracts low-level GPU programming and provides highly parallelized and optimized linear algebra algorithms \cite{malcolm2012arrayfire}.   

 
 We implemented Algorithm \ref{blassoAlg} using ArrayFire.  To solve $N$ Lasso problems of \eqref{eqn:l2l1_2} we implemented a generalized iterative re-weighted least-squares (GIRLS) \cite{bissantz2009convergence} algorithm. The GIRLS algorithm requires solving only least-squares sub-problems with linear algebra operations thus facilitating its implementation in ArrayFire. 






Figure \ref{fig:gputiming} shows execution times of computing a transport map with $N=500$ and PCE order of 3 running a Python implementation on an Intel Core i7 processor at 2.40 GHz(4 CPUs) and running with the ArrayFire implementation on an NVIDIA GeForce 840M GPU. As the complexity of the problem increases (determined by increasing $d$), the ArrayFire implementation readily outperforms the Python implementation.  This showcases the future possibilities for rapid computation of transport maps on architectures that feature parallelization capabilities.

\begin{figure}
\centering
\includegraphics[width=0.5\linewidth]{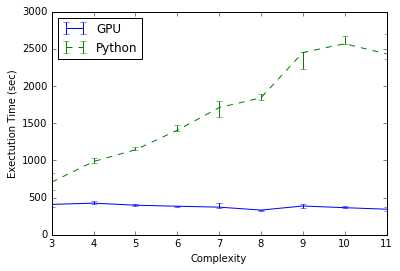}
\caption{Execution times for computing a transport map in Python and using a GPU. The horizontal axis represents the dimension of the latent random variable $X$.}
\label{fig:gputiming}
\end{figure}


\section{Discussion and Conclusion} \label{sec:discussion-conclusion}
We have shown that an i.i.d. posterior Bayesian Lasso sampler can be constructed with a measure transport framework with iteratively solving in a collection of standard LASSO problems in parallel and performing closed-form linear algebra updates.  
This formulation allows one to leverage the diversity of Lasso solvers to sample from posterior.  For example, we show how posterior Bayesian Lasso transport samplers can be constructed with a GPU. We also note that this algorithm could be readily implemented in other systems for parallelization such as cloud computing.   

Another potential application for inference with this transport-based approach is within the context of the Internet-of-Things (e.g. wearable electronics).  In these settings, energy efficiency is of paramount importance, and wireless transmission usually is the most energy-consuming.  Developing a framework such as ours where inference is performed on chip, obviates the need to transmit collected waveforms.  Instead, one only needs to transmit information about the posterior distribution, which is a sufficient statistic for any Bayesian decision making problem.  In our case, this boils down to transmission of coefficients of the polynomials representing the transport map.  From there, in the cloud for instance, i.i.d. prior samples may be transformed into i.i.d. posterior samples. Figure \ref{fig:analogtoinformation} shows a potential use of our posterior transmission scheme and a comparison to current transmission schemes.

\begin{figure}[htbp]
\setlength{\abovecaptionskip}{1pt plus 1pt minus 1pt}  
\centerline{\includegraphics[width=0.43\textwidth]{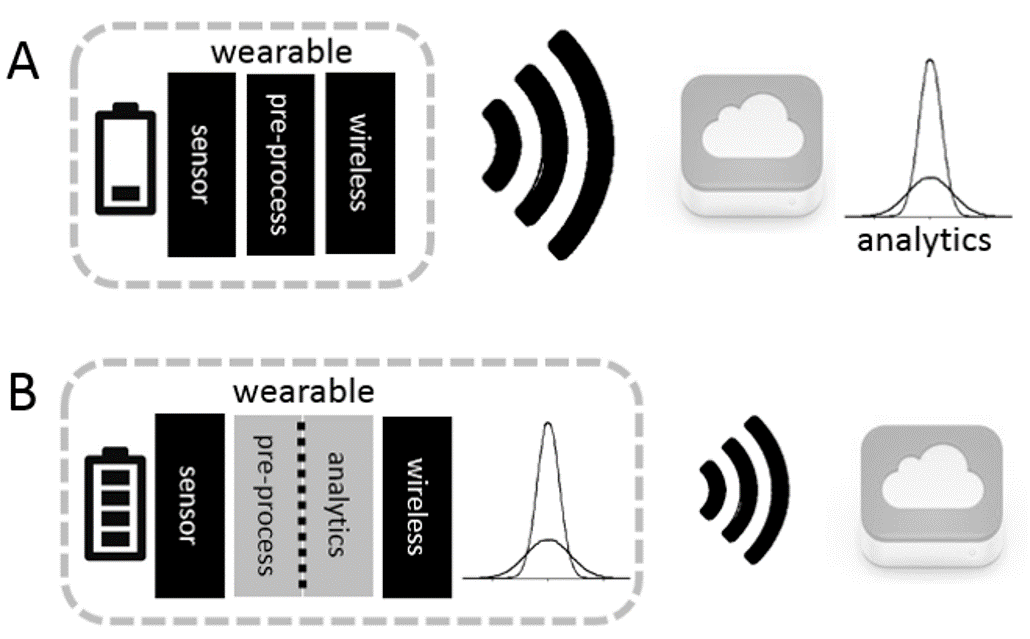}}
\caption{(A) shows conventional  wireless transmission schemes where signals are acquired and wirelessly transmitted; (B) shows our proposed scheme where inference is performed locally and only the posterior distribution is transmitted.  }
\label{fig:analogtoinformation}
\end{figure}

Another energy-efficient application could be achieved by implementing our Bayesian Lasso algorithm in analog systems. The Local Competitive Algorithm (LCA) first presented in \cite{rozell2008sparse} is an analog dynamical system inspired by neural image processing and exactly solves \eqref{eqn:l2l1}. This system has already been implemented in field-programmable analog arrays \cite{shapero2012low} and integrate-and-fire neurons \cite{gruenert2014understanding}, thus showing promising results for reduced energy in hardware implementations. 


In the LCA,  a set of parallel nodes, each associated with an element of the basis $\Phi_m \in \Phi$, compete with each other for representation of the input. The dynamics of LCA are expressed by a set of non-linear ordinary differential equations (ODEs) which represent simple analog components. The system's steady-state is the solution to \eqref{eqn:l2l1}.  Using the formulation presented in Theorem ~\ref{thm:BayesianLASSO:ADMM}, we could solve \eqref{eqn:ADMM:p} by presenting the LCA dynamics in terms of $\hat{y}$ and $\hat{\Phi}$: \begin{eqnarray*}
\dot{u}_m(t)&=&\frac{1}{\tau} \big[\langle \hat{\Phi}_m,  \hat{y} \rangle- u_m(t) - \sum\limits_{n \neq m}\langle \hat{\Phi}_m, \hat{\Phi}_n \rangle a_n(t) \big] \\
a_n(t) &\triangleq& T_{\lambda}(u_n(t)) =\max(0, u_n(t)-\lambda)
\end{eqnarray*}
where  $\hat{\Phi}_m$ denotes the $m$th column of $\hat{\Phi}$ and $T_{\lambda}$ is a thresholding function that induces local non-linear competition between nodes.

We have presented a framework to find a posterior for the Bayesian Lasso, however this parallelizable formulation could be easily extended to other $L_1$ priors and sparsity problems.  Dynamic formulations for spectrotemporal estimation of time series \cite{schamberg2017modularized} could be extended to a fully-Bayesian perspective to enable improved statistical inference and decision making.


\section*{Acknowledgments}
M. Mendoza acknowledges support from the NSF Graduate Research Fellowship.  T.P. Coleman acknowledges support from the NSF Center for Science of Information under Grant CCF-0939370, 
NSF grant IIS- 1522125, NIH grants 1R01MH110514, and ARO MURI grant ARO-W911NF-15-1-0479.

\bibliography{BayesLASSO}{}
\bibliographystyle{plain}

\end{document}